\documentclass[final,a4paper,leqno]{siamltex}

\usepackage{graphicx}
\graphicspath{{Pictures/}}
\usepackage{hyperref}

\usepackage{amsmath}
\usepackage{amssymb}
\def \v#1{\relax
          \ifmmode{\bf #1\/}\else{$\bf #1\/$}\fi}
\def \M#1{\relax
          \ifmmode{#1}\else{$#1$}\fi}
\def \affr #1{\relax
              \ifmmode{\cal #1}\else{$\cal #1$}\fi}
\def \koord #1{\relax
              \ifmmode{\hbox{$\sf #1$}}\else{$\sf #1$}\fi}
\newcommand{\I}{\hat{\imath}}

\newtheorem{remark}[theorem]{Remark}
\newtheorem{example}[theorem]{Example}

\title{General Midpoint Subdivision}

\author{Qi Chen
      \and Hartmut Prautzsch\thanks{Applied Geometry \& Computer Graphics (\url{http://geom.ibds.kit.edu}), Karlsruhe Institute of Technology, Am Fasanengarten 5, 76131 Karlsruhe, Germany ({\tt chenqi@ira.uka.de} and {\tt prautzsch@kit.edu}).}}

\begin{document}

\maketitle

\begin{abstract}
In this paper, we introduce two generalizations of midpoint subdivision and analyze the smoothness of the resulting subdivision surfaces at regular and extraordinary points.

The smoothing operators used in midpoint and mid-edge subdivision connect the midpoints of adjacent faces or of adjacent edges, respectively. An arbitrary combination of these two operators and the refinement operator that splits each face with $m$ vertices into $m$ quadrilateral subfaces forms a \emph{general midpoint subdivision operator}.
We analyze the smoothness of the resulting subdivision surfaces by estimating the norm of a special second order difference scheme and by using established methods for analyzing midpoint subdivision. The surfaces are smooth at  their regular points and they are also smooth at extraordinary points for a certain subclass of general midpoint subdivision schemes.

Generalizing the smoothing rules of non general midpoint subdivision schemes around extraordinary and regular vertices or faces results in a class of subdivision schemes, which includes the Catmull-Clark algorithm with restricted parameters. We call these subdivision schemes \emph{generalized Catmull-Clark schemes} and we analyze their smoothness properties.
\end{abstract}



\pagestyle{myheadings}
\thispagestyle{plain}
\markboth{\uppercase{General Midpoint Subdivision}}{\uppercase{Qi Chen and Hartmut Prautzsch}}

\section{Introduction}

The Lane-Riesenfeld algorithm \cite{LR80} for subdividing cardinal spline curves and tensor product surfaces has been fundamental to the development of box spline and general stationary subdivision schemes \cite{Prautzsch84,CLR84,DM84,DGL91,CDM91}. The algorithm is quite simple since it is based on iteratively computing midpoints of the edges and faces of regular quadrilateral meshes. It can easily be generalized to arbitrary meshes and this is now called \emph{midpoint subdivision}. It includes particular instances of the Catmull-Clark and Doo-Sabin schemes \cite{CC78,DS78}, which even predate the Lane-Riesenfeld algorithm.

A midpoint subdivision scheme consists of an operator $A^{n-1}R$ of degree $n\in \mathbb N$, which is used successively to subdivide an input mesh $\affr M$. 
The \emph{refinement operator} $R$ maps $\affr M$ to the quadrilateral mesh $R\affr M$ by splitting every face of $\affr M$ at its center as shown at the top of Figure~\ref{FIG:OperatorsRAV}.
The \emph{averaging operator} $A$ maps $\affr M$ to the dual mesh $A\affr M$ that connects the centers of adjacent faces as shown in the middle of Figure~\ref{FIG:OperatorsRAV}.

Lane and Riesenfeld developed their algorithm to generate known curves and surfaces but the limiting surfaces of the midpoint subdivision schemes are not explicitly known, which complicates their analysis. In fact, Catmull had no proof for his scheme, which he presented in his dissertation \cite{Catmull74} and which he published later together with Clark \cite{CC78}. The first analysis was outlined by Sabin in 1978 \cite{DS78} but it took another 15 years until Reif established a rigorous proof \cite{Reif93}. Reif's proof is based on an analysis of the subdivision matrix and requires (numerically) computing or estimating its spectral properties. The numerical nature of Reif's proof is a bottleneck when it is used to analyze infinite classes of subdivision schemes. Therefore, Zorin and Schr\"oder were merely able to analyze the midpoint schemes up to degree $9$ \cite{ZS2001} and only recently further geometric arguments were developed that helped to show that all midpoint subdivision surfaces of any arbitrarily high degree are $C^1$ everywhere \cite{PC2011}. 

Many arguments used in \cite{PC2011} also apply to other subdivision schemes based on convex combinations and further generalizations of the results are possible. However, certain technical difficulties and some fundamental challenges remain that are unsolved. In this paper, we address some of them and enlarge the class of subdivision schemes which can be analyzed, after appropriate extensions, by the technique developed in \cite{PC2011}. Although we add to the averaging and refinement operators ``only'' the edge midpoint operator, which is simpler than the other two, we are challenged by new problems:
\begin{itemize}
\item
The limiting surfaces of regular meshes are only known for midpoint subdivision and mid-edge subdivision. Again, if we had only finitely many schemes for which we could compute their symbols and run the standard analysis, this would not be a problem. In this paper, however, we consider infinitely many subdivision operators that can be decomposed into the three basic operators mentioned above.
\item
The mid-edge operator changes the ``orientation'' of a mesh and therefore, we face additional topological difficulties.
\item
The eigenvectors of a subdivision matrix have a basis of rotation symmetric eigenvectors of different frequencies. For midpoint subdivision, the eigenvectors with frequency 0 are identical for all valencies of the extraordinary point but this is not the case for the general midpoint schemes we consider in this paper. This is particularly a problem in Section~\ref{SECTION:DefGenMidSubII}, where we work with convex combinations other than midpoint averaging.
\end{itemize}

General midpoint subdivision is defined in Section~\ref{SECTION:DefGenMidSub} and its $C^1$-property for regular meshes is shown in Sections~\ref{SECTION:DifferenceSchemes} and \ref{SECTION:Smoothness}. For arbitrary meshes, the subdivision is shown to be convergent in Section~\ref{SECTION:BasicObservations} and its $C^1$-property is analyzed in Sections~\ref{SECTION:ReviewGeometricMethod2}, \ref{SECTION:ChapMap}, and \ref{SECTION:GenMidC1}. Some examples of general midpoint subdivision schemes are given in Section~\ref{SECTION:Examples} and another generalization of midpoint subdivision, \emph{generalized Catmull-Clark subdivision}, is defined and analyzed in Section~\ref{SECTION:DefGenMidSubII}.
We provide a few concluding remarks in Section~\ref{SECTION:Conclusion}.

\section{General midpoint subdivision}\label{SECTION:DefGenMidSub}

First, we consider a generalization of midpoint subdivision including also the \emph{mid-edge or simplest subdivision operator} $V$. This operator maps a mesh $\affr M$ to the mesh $V\affr M$ that connects the midpoints of adjacent edges of $\affr M$ as illustrated at the bottom of Figure~\ref{FIG:OperatorsRAV}. 

\begin{figure}[htb]
\begin{center}
{\includegraphics{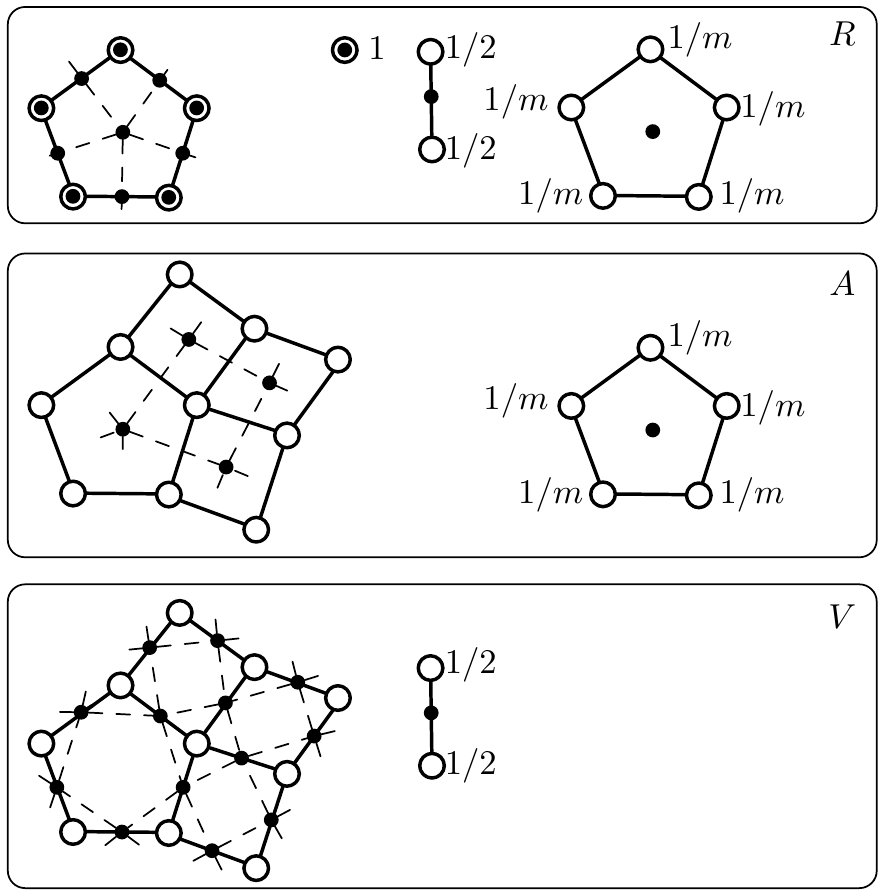}}
\caption{Three basic operators for subdividing quadrilateral meshes, where $m$ ($=$ number of adjacent or boundary edges) is the valence of a vertex or face: refinement operator $R$ (top), averaging operator $A$ (middle), and mid-edge or simplest operator $V$ (bottom).}
\label{FIG:OperatorsRAV}
\end{center}
\end{figure}

More precisely, we consider the \emph{general midpoint subdivision operators}
\[U=A^{a_s} V^{v_s} R^{r_s} A^{a_{s-1}} V^{v_{s-1}} R^{r_{s-1}}  \cdots A^{a_1} V^{v_1} R^{r_1}\]
for
\[a_1, v_1, r_1, \ldots, a_s, v_s, r_s \ge 0\]
with $a+v\ge 1$ and $v+r\ge 1$, where $a=\sum_{i=1}^s a_i, \; v=\sum_{i=1}^s v_i$, and $r=\sum_{i=1}^s r_i$.

The midpoint subdivision schemes $A^{n-1} R$ of degrees $n\ge 2$ and the mid-edge subdivision scheme $V$ are special general midpoint subdivision schemes. 

We show that general midpoint subdivision surfaces are $C^1$ at their regular points and that they are also $C^1$ at their extraordinary points for a certain subclass of general midpoint subdivision schemes.

\section{A second order difference scheme}\label{SECTION:DifferenceSchemes}
To analyze the smoothness of general midpoint subdivision schemes for regular (quadrilateral) meshes, we introduce a special second order difference scheme in this section and analyze its norm in the next section.

A regular mesh $\affr C$ can be represented by the biinfinite matrix $\affr C = [\v c_{\v i}]_{\v i \in \mathbb Z^2}$ of its vertices $\v c_{\v i}$, which are connected by the edges $\v c_{\v j}\v c_{\v j +\v e_1}, \v c_{\v j}\v c_{\v j +\v e_2},\; \v j\in \mathbb Z^2$,
where
\[[\v e_1\;\v e_2\;\v e_3\;\v e_4] = \left[ \begin{array}{rrrr} 1 & 0 & 1 & -1\\ 0 & 1 & 1 & 1 \end{array}\right]\;.\]

To analyze smoothness, we need the (backward) differences
\[\nabla_k \v c_{\v i} = \v c_{\v i} - \v c_{\v i-\v e_k}\;,\quad k=1, 2, 3, 4,\]
and the mesh $\affr C_{\widehat\nabla \nabla} := \left[\widehat\nabla \nabla \v c_{\v i}\right]_{\v i \in \mathbb Z^2}$ of the second order differences
\[\widehat\nabla \nabla \v c_{\v i} := [\nabla_1 \nabla_1 \v c_{\v i}\quad \nabla_3 \nabla_1 \v c_{\v i}\quad \nabla_4 \nabla_1 \v c_{\v i}\quad \nabla_2 \nabla_2 \v c_{\v i}\quad \nabla_3 \nabla_2 \v c_{\v i}\quad \nabla_4 \nabla_2 \v c_{\v i}]\;,\]
where $\nabla_j \nabla_k \v c_{\v i}  = \nabla_j (\nabla_k \v c_{\v i})$.

Let $U$ be $R$, $A$, $V$, or a general midpoint subdivision operator. Since $U$ maps a linear mesh $[l(\v i)]_{\v i \in \mathbb Z^2}$ onto itself, where $l: \mathbb R^2 \to \mathbb R$ is any linear function, a second order difference scheme $U_{\widehat\nabla \nabla}$ exists such that
\begin{equation}\label{EQ:Diff2}
(U \affr C)_{\widehat\nabla \nabla} = U_{\widehat\nabla \nabla} \affr C_{\widehat\nabla \nabla}
\end{equation}
for all meshes $\affr C$ (see \cite[Equations (11) and (12)]{Kobbelt2000}). 

Furthermore, for a mesh $\affr C = [\v c_{\v i}]_{\v i \in \mathbb Z^2}$, each element of $\{ \nabla_j \nabla_k \v c_{\v i} \;| \; j, k \in \{1, 2\},  \v i \in \mathbb Z^2\}$ is a linear combination of $\affr C_{\widehat\nabla \nabla}$ because 
\[\nabla_1 \nabla_2 \v c_{\v i}  = \nabla_2 \nabla_1 \v c_{\v i} = \nabla_3 \nabla_1 \v c_{\v i}- \nabla_1 \nabla_1 \v c_{\v i - \v e_2}\;.\]

\section{Smoothness for regular meshes}\label{SECTION:Smoothness}
Applying a general midpoint subdivision operator
\[U=A^{a_s} V^{v_s} R^{r_s} \cdots A^{a_1} V^{v_1} R^{r_1}\]
to the grid $\mathbb Z^2$, we obtain the scaled grid $\mathbb Z^2\, /\, 2^{r+v/2}$ or a rotated version of it, where $r=\sum_i r_i$ and $v=\sum_i v_i$. 
We denote the \emph{scaling factor} by
\begin{equation}\label{EQ:SCALING}
\sigma = \sigma(U) = 2^{-r-v/2}
\end{equation}
and recall the following well-known fact, which can be derived easily from  \cite[Theorem 7.6]{Dyn92}.
\begin{theorem}[{\rm Smoothness condition for regular control meshes}]
\label{SATZ:C1Condition}
Let
\[\Vert \affr C\Vert_\infty := \sup_{\v i\in \mathbb Z^2} \Vert\v c_{\v i}\Vert\quad \mbox{and} \quad
\Vert U\Vert := \sup_{\Vert \affr C\Vert_\infty=1} \Vert U \affr C\Vert_\infty\;.\]
If $U$ maps $\mathbb Z^2$ to $\sigma\, \mathbb Z^2$ and if $U_{\widehat\nabla \nabla}/\sigma$ has some contractive power, i.\,e., if
\[\left\Vert U^k_{\widehat\nabla \nabla} \right\Vert < \sigma^k\]
for some $k$, then $U$ is a \emph{$C^1$-scheme}, i.\,e., for any bounded mesh $\affr C$ there is a continuously differentiable function $\v f(x, y)$ such that
\[\lim_{k\to\infty} \left\Vert U^k \affr C - \v f\left(\sigma^{k} \, \mathbb Z^2\right)\right\Vert = 0\;.\]
\end{theorem}

To check the prerequisites of this theorem, we use
\begin{lemma}[{\rm Estimates about the difference schemes}]
\label{LEMMA:Estimates1}
\begin{enumerate}
\item[(a)] $\Vert A_{\widehat\nabla \nabla}\Vert=1$\;,
\item[(b)] $\Vert R_{\widehat\nabla \nabla}\Vert=1/2$\;,
\item[(c)] $\Vert V_{\widehat\nabla \nabla}\Vert\le1/2$\;, \mbox{and}
\item[(d)] $\Vert A_{\widehat\nabla \nabla}R_{\widehat\nabla \nabla}\Vert \le 3/8$.
\end{enumerate}
\end{lemma}
\begin{proof}
{
Since $A_{\widehat\nabla \nabla} = A$ and $\Vert A\Vert=1$, we obtain (a).

\begin{figure}[htb]
\begin{center}
{\includegraphics{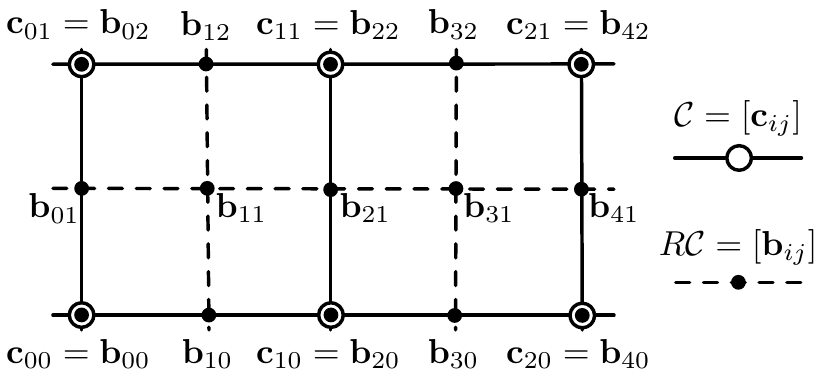}}
\caption{A pair of meshes $\cal C$ and $R \cal C$.}
\label{FIG:OperatorsR6}
\end{center}
\end{figure}

Figure~\ref{FIG:OperatorsR6} shows a pair of meshes $\affr C$ and $R\affr C$ schematically. For these meshes, we have
\begin{eqnarray*}
\nabla^2_1\v b_{20} &=& \nabla^2_1\v b_{21} = 0\;,\\
\nabla^2_1\v b_{30} &=& \frac{1}{2} \nabla^2_1\v c_{20}\;,\\
\nabla^2_1\v b_{31} &=& \frac{1}{4}\left(\nabla^2_1\v c_{20}+ \nabla^2_1\v c_{21}\right)\;,\\\nabla_3\nabla_1\v b_{21}&=&\frac{1}{4}\nabla_2\nabla_1\v c_{11} = \frac{1}{4}(\nabla_3\nabla_1\v c_{21}-\nabla_1\nabla_1\v c_{21})\;,\quad \mbox{and}\\
\nabla_3\nabla_1\v b_{31}&=& \frac{1}{4}\left(\nabla_3\nabla_1\v c_{21}+\nabla^2_1\v c_{20}\right)\;.
\end{eqnarray*}
Similarly, we can derive such equalities for all other elements in $R_{\widehat\nabla \nabla}\affr C_{\widehat\nabla \nabla}$. This proves (b).

\begin{figure}[htb]
\begin{center}
{\includegraphics{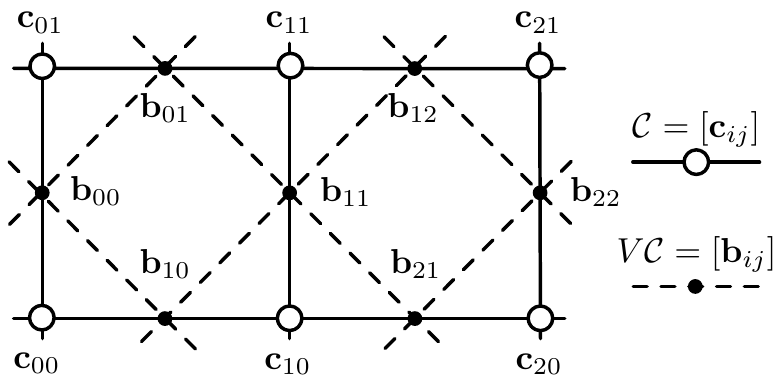}}
\caption{A pair of meshes $\cal C$ and $V \cal C$.}
\label{FIG:OperatorsV4}
\end{center}
\end{figure}

Using the notation shown schematically in Figure~\ref{FIG:OperatorsV4} for a pair of meshes $\affr C$ and $V \affr C$, (c) follows from
\[\nabla^2_2\v b_{12} = \nabla_3\nabla_2 \v b_{12} = \nabla_3\nabla_2 \v b_{22} =  \frac{1}{2} \nabla_3\nabla_1 \v c_{21}\]
and from analogous equalities for all other elements in $V_{\widehat\nabla \nabla}\affr C_{\widehat\nabla \nabla}$.

\begin{figure}[htb]
\begin{center}
{\includegraphics{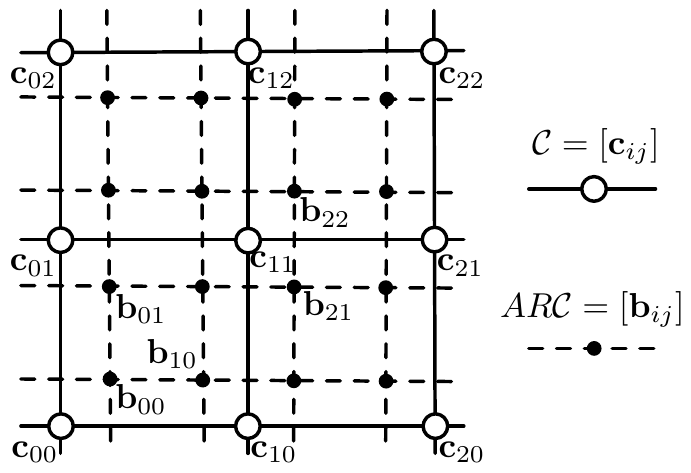}}
\caption{A pair of meshes $\cal C$ and $AR \cal C$.}
\label{FIG:OperatorsR7}
\end{center}
\end{figure}
Finally, for the meshes $\affr C$ and $AR \affr C$ in Figure~\ref{FIG:OperatorsR7}, we have
\begin{eqnarray*}
\nabla^2_1\v b_{21} &=& \frac{1}{16}\left(3\nabla^2_1\v c_{21} + \nabla^2_1\v c_{20}\right)\;,\\
\nabla_3\nabla_1\v b_{21}&=&\frac{1}{16}(3\nabla_1\v c_{21}+\nabla_1\v c_{11}+\nabla_1\v c_{20}-5\nabla_1\v c_{10}) \\
&=& \frac{1}{16}\left(3\nabla_3\nabla_1\v c_{21}+\nabla_4\nabla_1\v c_{11} +2\nabla_1^2 \v c_{20}\right)\;,\\
\nabla_3\nabla_1\v b_{22}&=&\frac{1}{16}(\nabla_1\v c_{22}+\nabla_1\v c_{12}+3\nabla_1\v c_{21}-3\nabla_1\v c_{11}-2\nabla_1\v c_{10}) \\
&=& \frac{1}{16}\left(\nabla_3\nabla_1\v c_{22}+\nabla_4\nabla_1\v c_{12} +2\nabla_1^2 \v c_{21}+2\nabla_3\nabla_1\v c_{21}\right)\;,
\end{eqnarray*}
and analogous equalities for all other elements in $A_{\widehat\nabla \nabla}R_{\widehat\nabla \nabla}\affr C_{\widehat\nabla \nabla}$. Hence,
\[\left\Vert A_{\widehat\nabla \nabla}R_{\widehat\nabla \nabla}\right\Vert
\le \max\left\{\frac{3+1}{16}, \frac{3+1+2}{16}, \frac{1+1+2+2}{16}\right\} = \frac{3}{8}\;,\]
which establishes (d).
}
\qquad
\end{proof}

\begin{theorem}[{\rm $C^{1}$ continuity for regular meshes}]
\label{SATZ:C1REGULAR}
Any general midpoint scheme $U=A^{a_s} V^{v_s} R^{r_s} \cdots A^{a_1} V^{v_1} R^{r_1}$ with $a+v\ge 1$ and $v+r\ge 1$ is a $C^1$-scheme for regular meshes, where $a = \sum_{i=1}^s a_i$, $v =  \sum_{i=1}^s v_i$, and $r=\sum_{i=1}^s r_i$.
\end{theorem}
\begin{proof}
{
Using Equation~(\ref{EQ:Diff2}) and Lemma~\ref{LEMMA:Estimates1}, we obtain
\begin{eqnarray*}
\left\Vert \left(U^2\right)_{\widehat\nabla \nabla}\right\Vert 
&=& \left\Vert A_{\widehat\nabla \nabla}^{a_s} \cdots  R_{\widehat\nabla \nabla}^{r_1} A_{\widehat\nabla \nabla}^{a_s} \cdots  R_{\widehat\nabla \nabla}^{r_1} \right\Vert\\
&\le& \left\{ \begin{array}{ll}2^{-2v-2r} & \mbox{if $v\ge 1$}\\
\frac{3}{4}\,2^{-2r} & \mbox{if $v=0$ and $a, r\ge 1$}\end{array}\right.\\
&<& 2^{-v-2r}=\sigma\left(U^2\right)\;.
\end{eqnarray*}
Since $U^2$ maps $\mathbb Z^2$ to $\sigma\left(U^2\right) \, \mathbb Z^2$, we conclude from Theorem~\ref{SATZ:C1Condition} that $U^2$ is a $C^1$-scheme. Since $U$ and $U^2$ have the same limiting surfaces, $U$ is also a $C^1$-scheme.
}
\qquad
\end{proof}

\section{Basic observations}\label{SECTION:BasicObservations}
In this section, we consider general midpoint subdivision for arbitrary (quadrilateral) meshes with extraordinary vertices or faces and we show that general midpoint subdivision converges, i.\,e., general midpoint subdivision surfaces are continuous.

Interior faces and vertices of a (quadrilateral)  mesh are called \emph{extraordinary} if their valence does not equal $4$.
Subdividing by $V$, $R$, and $A$ does not increase the number of extraordinary elements and isolates these elements. Therefore, it suffices to consider only (sub)meshes with one extraordinary element, as illustrated in Figure~\ref{Abb:BSP_Ring_Netz}. These meshes are called \emph{ringnets}.

\begin{figure}[htb]
\begin{center}
{\includegraphics{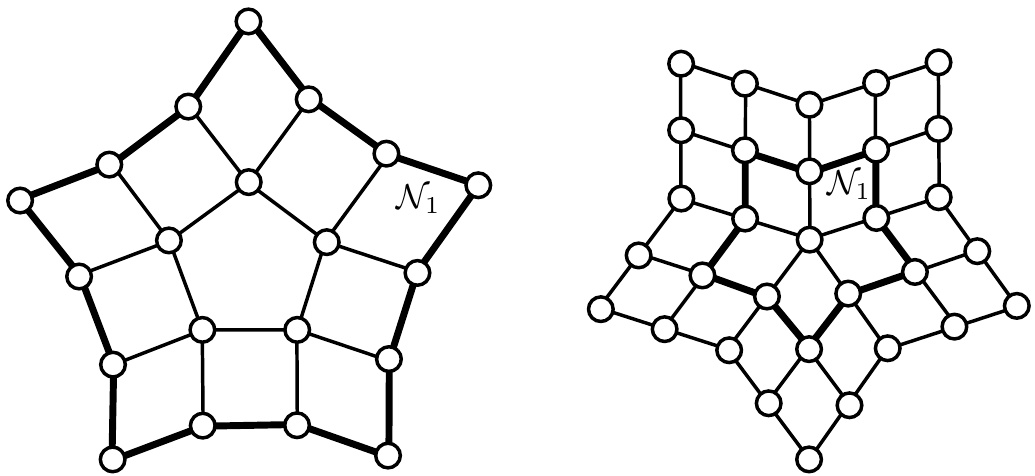}}
\caption{Examples of rings and ringnets: a $1$-ringnet with an extraordinary face of valence $5$ (left) and  a $2$-ringnet with an extraordinary vertex of valence $5$ (right). The first rings $\affr N_1$ in both meshes are marked by bold lines.}
\label{Abb:BSP_Ring_Netz}
\end{center}
\end{figure}

Given a ringnet $\affr N$ and a general midpoint subdivision operator 
\[U=A^{a_s} V^{v_s} R^{r_s} A^{a_{s-1}} V^{v_{s-1}} R^{r_{s-1}}  \cdots A^{a_1} V^{v_1} R^{r_1}\;,\]
we generate the sequence $\affr N^{(k)} = U^k \affr N$. Furthermore, we subdivide just the regular parts of any $\affr N^{(k)}$ and obtain for every $k$ a limiting surface $\v s_k$. Since $\v s_k$ is part of $\v s_{k+1}$, it suffices to study the operator $U^2$ instead of $U$ and we assume without loss of generality that $v=\sum_{i=1}^s v_i$ is even. Under this assumption, the ``orientation'' of $U \affr N$ does not change from that of $\affr N$ and the meshes $\affr N^{(k)}$ are either all primal (i.\,e., have no extraordinary face) or all dual (i.\,e., have no extraordinary vertex).

We say that a mesh $\affr N$ \emph{influences} another subdivided mesh $\affr M$ if during subdivision every vertex in $\affr N$ has an effect on some vertex in $\affr M$ and if additionally all vertices in $\affr M$ depend on $\affr N$.

\begin{definition}[{\rm Ring, ringnet, core (mesh)}]
\label{DEF:Omega}
Let $\affr N_0^{(k)}$ be the subnet of $\affr N^{(k)}$ consisting of the extraordinary vertex or face of $\affr N^{(k)}$.
The $k$-th ring around $\affr N_0$ is denoted by  $\affr N_k$. 
The mesh $\affr N_{0\ldots k}$ consists of $\affr N_0$ and the next $k$ rings of vertices around $\affr N_0$ and is called a \emph{$k$-ringnet} or \emph{$k$-net} for short.  
Furthermore, the submesh $\affr N_{i\ldots j}$ consists of $\affr N_i$, \ldots, $\affr N_j$.

The \emph{core} or \emph{core mesh} of $\affr N$ (with respect to $U$) consists of  all vertices influencing $\affr N_0^{(k)}$ for some $k\ge 1$. We denote it by $\affr N_{c}$. Figure~\ref{FIG:Omega} shows some examples.
The $k$-th ring around $\affr N_c$ is denoted by $\affr N_{c.k}$.
The subnet $\affr N_{c.0\ldots k}$ consists of $\affr N_c$ and the next $k$ rings of vertices around $\affr N_c$ and is called a \emph{$c.k$-ringnet} or \emph{$c.k$-net} for short.
Furthermore, the submesh $\affr N_{c.i\ldots j}$ consists of $\affr N_{c.i}$, \ldots, $\affr N_{c.j}$ if $1 \le i \le j$.
\end{definition}

\begin{figure}[htb]
\begin{center}
{\includegraphics{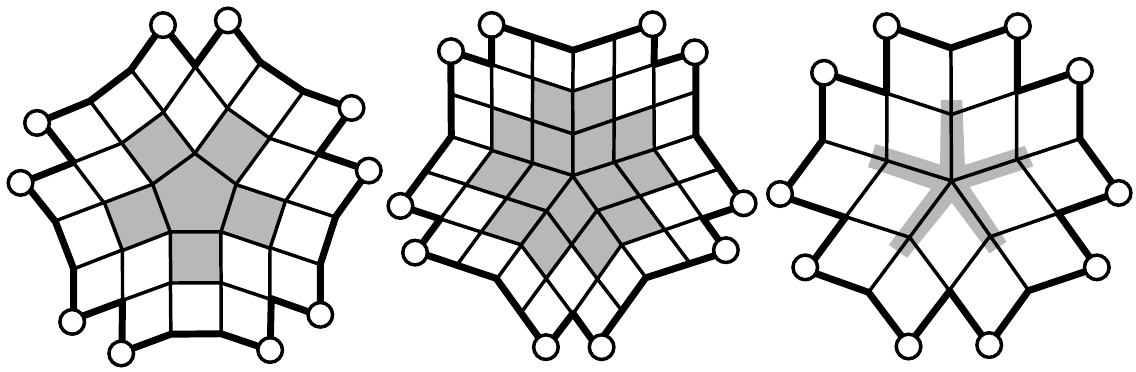}}
\caption{ The core meshes (in gray) and the rings $\affr N_{c.1}$ (bold edges) of a ringnet with an extraordinary face or vertex of valence $5$ for the subdivision schemes $\mathit{VAV}$ (left), $\mathit{AVAV}$ (middle), and ${RV^2R}$ (right). The convex corners of $\affr N_{c.1}$ are marked by hollow dots.}
\label{FIG:Omega}
\end{center}
\end{figure}

Depending on the context, we treat any mesh as a matrix whose rows represent the vertices or as the set of all vertices.

It is straightforward to prove
\begin{lemma}[{\rm Dependence between nets after a subdivision step}]
\label{LEMMA:DependenceU}
\begin{enumerate}
\item[(a)] The subnet of $A \affr N$, $R \affr N$, or $V \affr N$ consisting of all vertices depending on any connected subnet of $\affr N$ is connected.
\item[(b)] $\affr N_{c.0\ldots k}$ determines $\affr N^{(1)}_{c.0\ldots 2k}$ for $k\ge 0$, i.\,e.,
\[\affr N^{(1)}_{c.0\ldots 2k} = (U\;\affr N_{c.0\ldots k})_{c.0\ldots 2k}\;.\]
\end{enumerate}
\end{lemma}

\begin{lemma}[{\rm Dependence property of a core mesh}]
\label{LEMMA:EigenschaftenVonOmega1Netz2}
Let $r\ge 0$. Then, for some constant  $q$ depending on $U$ and $r$, every vertex in $\affr N_c$ influences all vertices in $\affr N^{(q+k)}_{c.0\ldots r}=(U^{q+k}\;\affr N)_{c.0\ldots r}$ for all $k\geq 0$,  which is denoted by 
\[\affr N_{c} \Rrightarrow \affr N^{(q+k)}_{c.0\ldots r}\;.\]
\end{lemma}
\begin{proof}
{
Let $\alpha$ be such that $\affr N_{0\ldots \alpha}$ contains $\affr N_{c.0\ldots r}$. For sufficiently large $l$ and any $k\ge 0$, every vertex in $\affr N_{c}$ influences all vertices in $\affr N_0^{(l+k)}$, all vertices in $\affr N_{0\ldots 1}^{(l+k+1)}$, \ldots, and all vertices in $\affr N_{0\ldots \alpha}^{(l+k+\alpha)}$ and hence also all vertices in $\affr N_{c.0\ldots r}^{(l+k+\alpha)}$. Thus, we obtain the lemma with $q = l+\alpha$.
}
\qquad
\end{proof}

As explained above, every ringnet $\affr N^{(k)}$ defines a surface $\v s_k$. Since $\v s_{k+1}$ contains $\v s_k$, we can consider the difference surface $\v r_k = \v s_{k+1} \backslash \v s_k$ whose control points are contained in a sufficiently large subnet $\affr N^{(k)}_{c.0\ldots\rho}$ with $\rho$ not depending on $k$.
Due to Lemma~\ref{LEMMA:DependenceU} (b), the operator $U$ restricted to $c.\rho$-nets can be represented by a stochastic matrix $S$ called the \emph{subdivision matrix}, i.\,e.,
\begin{equation}\label{EQ:SubMatrix}
\affr N_{c.0\ldots\rho}^{(k+1)} =  S\,\affr N_{c.0\ldots\rho}^{(k)}\;.
\end{equation}

\begin{theorem}[{\rm $C^0$-property of $U$}]
\label{SATZ:C0}
The subdivision surfaces generated by $U$ are $C^0$ continuous.  
\end{theorem}
\begin{proof}
{
Since the subdivision matrix $S$ is stochastic, i.\,e., $S$ is a non-negative and real matrix and each row of $S$ sums to $1$, $1$ is the dominant eigenvalue of $S$.  
Due to Lemma~\ref{LEMMA:EigenschaftenVonOmega1Netz2}, there is an integer $l\ge 1$ such that 
\[\affr N_{c} \Rrightarrow \affr N^{l}_{c.0\ldots \rho} = S^l\, \affr N_{c.0\ldots \rho} \;.\]
This implies that $S^l$ has a positive column and, according to \cite[Theorem 2.1]{MP89a}, any sequence $(S^i \, \v c)$ converges to a multiple of the vector $[1\,\ldots\, 1]^\mathrm{t}$ as $i \to \infty$ for all real vectors $\v c$.
Therefore, the only dominant eigenvalue of $S$ is $1$ and it has algebraic multiplicity $1$.

Hence, the difference surfaces $\v s_i \backslash \v s_{i-1}$ converge to a point and the surfaces generated by $U$ are continuous.
}
\quad
\end{proof}

To analyze the spectrum of the subdivision matrix $S$, we order any $c.\rho$-net $\affr N$  such that
\[\affr N = \left[ \begin{array}{l} \affr N_{c} \\ \affr N_b \\ \affr N_a \\ \affr N_{c.2} \\ \vdots \\ \affr N_{c.\rho}\end{array} \right]\;,\]
where $\affr N_a$ consists of the convex corners and $\affr N_b$ of all other points in $\affr N_{c.1}$ (see Figure~\ref{FIG:Omega} for an illustration of the convex corners). 
With this arrangement, the subdivision matrix $S$ has the lower triangular form
\[S = \left[ 
\begin{array}{cccccc} \M C\\ * & \M B\\ * & * & \M A \\  * & * & * & 0\\
\vdots & & & \ddots & \ddots\\
  * & \hdots & \hdots & \hdots & * & 0
\end{array}
\right]\;,\]
where
\begin{eqnarray}
\affr N_c^{(1)} &=& C \,\affr N_c\;, \label{EQ:C}\\
\affr N_b^{(1)} &=&  \left[ \begin{array}{cc} * & B \end{array} \right]
\left[ \begin{array}{l} \affr N_{c} \\ \affr N_b \end{array} \right],\quad \mbox{and} \label{EQ:B}\\
\affr N_a^{(1)} &=& \left[ \begin{array}{ccc} * & * & A \end{array} \right] 
\left[ \begin{array}{l} \affr N_c \\ \affr N_b \\ \affr N_a \end{array} \right]\;.\label{EQ:A}
\end{eqnarray}

To verify this, we recall that any point influencing the core mesh influences $\affr N_0$ and thus belongs to the core mesh. This implies Equation~(\ref{EQ:C}) and shows that $\affr N_{c.1}$ influences only points in $(U\affr N)_{c.1\ldots \infty}$. Hence, $\affr N_{c.2}$ influences only points in $(U\affr N)_{c.2\ldots \infty}$, etc. 
We observe that $\affr N_a$ does not influence any point in $\affr N_b^{(1)}$ and if $U$ has a factor $V$, then $\affr N_a$ does not even influence $\affr N_a^{(1)}$ (see  Figure~\ref{FIG:B1B2}). Thus, Equations~(\ref{EQ:B}) and (\ref{EQ:A}) follow.
Moreover, due to Lemma~\ref{LEMMA:DependenceU}\,(b), $(U \affr N)_{c.2}$ is determined by $(U \affr N)_{c.0\ldots 1}$, and $(U \affr N)_{c.3}$ is determined by $(U \affr N)_{c.0\ldots 2}$, etc. This implies the zero blocks in $S$.

Hence, the eigenvalues of $S$ are zero or are the eigenvalues of the blocks $\M C$, $\M B$, and  $\M A$.

\begin{lemma}[{\rm Spectral radii of $B$ and $A$}]
\label{LEMMA:ZeilensummeVonB}
The spectral radii $\rho_{B}$ and $\rho_{A}$ of $B$ and $A$ satisfy 
\[\rho_{B} \leq 2^{-r-a-v} \quad \mbox{and} \quad \rho_{A} \leq
\left\{ \begin{array}{lll} 4^{-r-a} &\mbox{if}& v = 0 \\ 0 &\mbox{if}& v>0 \end{array} \right.\;.\]
In particular,  $\rho_B, \rho_A \le \sigma(U)=2^{-r-v/2}$.
\end{lemma}
\begin{proof}
Since $B$ is non-negative, we get \cite[Corollary 6.1.5 on Page 346]{HJ85}
\[\rho_{B} \le \Vert B \Vert_\infty = \Vert B \v 1\Vert_\infty, \quad \mbox{where} \quad \v 1 := [1\;\ldots\;1]^{\mathrm{t}}\;.\]
The vector $B \v 1$ represents $\affr N_{c.1}^{(1)}$ without the convex corners if $\affr N_{c} = 0$, $\affr N_b = 1$, $\affr N_a = 0$, and $\affr N_{c.2\ldots\rho} = 0$. 
Due to Lemma~\ref{LEMMA:DependenceU} (a), the vertices influenced by $\affr N_{c.1}$ under $A$, $R$, and $V$ form a connected subnet of $A \affr N$, $R \affr N$, and $V \affr N$ respectively with an inner boundary whose vertices that do not form a convex corner have a value $\le 1/2$. By induction, $\rho_{B}\le 2^{-r-a-v}$ follows.

Similarly, we get $\rho_{A} \le 4^{-r-a}$ if $v=0$.

If $v>0$, then decompose $U$ into $U_1VU_2$ where $U_2$ is the identity operator or a sequence of the operators $A$ and $R$. 
Due to Lemma~\ref{LEMMA:DependenceU} (a), the vertices influenced by $\affr N_{c.1}$ form band rings with inner boundaries $\affr B_1$ in $U_2 \affr N$ and $\affr B_2$ in $VU_2\affr N$. For $\affr N_{c} = 0$, we get that $\affr N_{c.1}^{(1)}$ is determined by $\affr B_2$ which is determined by $\affr B_1$ excluding its convex corners (see Figure~\ref{FIG:B1B2}). Since the convex corners of $\affr N_{c.1}$ have no influence on $\affr B_1$ excluding its convex corners, we obtain $\rho_{A} = 0$.
\begin{figure}[htb]
\begin{center}
{\includegraphics{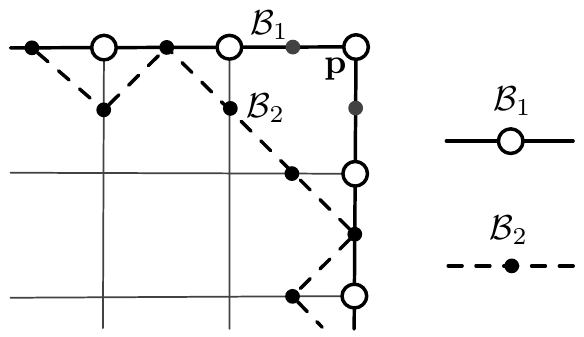}}
\caption{${\cal B}_1$ and ${\cal B}_2$ are the inner boundaries of the band rings in $U_2 {\cal N}$ and  in $VU_2{\cal N}$ that are influenced by ${\cal N}_{c.1}$. Note that the convex corner $\mathbf p$ of ${\cal B}_1$ has no influence on $\affr B_2$.}
\label{FIG:B1B2}
\end{center}
\end{figure}
\qquad
\end{proof}

\section{Symmetric ringnets}\label{SECTION:ReviewGeometricMethod2}

For the $C^1$ analysis of general midpoint subdivision, we need to investigate the eigenvectors and eigenvalues of the subdivision matrix $S$. We do this by subdividing special grid meshes as in \cite{PC2011} and recall the basic definitions in this and the next section.

\begin{definition}[{\rm Grid mesh}]
\label{DEF:Netzstruktur}
A \emph{primal grid mesh} of valence~$m$ and frequency~$f$ is a  planar primal ringnet with the vertices  
\[\v g_{ij}^l = \left[\begin{array}{l} \mathrm{Re}(g_{ij}^l)\\ \mathrm{Im}(g_{ij}^l)\end{array} \right]  \in \mathbb R^2\;,\]
where $g_{ij}^l = i e^{\I 2 \pi lf/m} + j e^{\I 2 \pi (l+1)f/m} \in \mathbb C$ and $i,j\geq 0,\; l\in \mathbb Z_m,\; \I = \sqrt{-1}$.\\
A \emph{dual grid mesh} of valence~$m$ and frequency~$f$ consists of the vertices
\[\v h_{ij}^l = \frac{1}{4}(\v g_{i-1,j-1}^{l} + \v g_{i,j-1}^{l} + \v g_{i-1,j}^{l} + \v g_{i,j}^{l}), \quad i,j\geq 1,\; l\in \mathbb Z_m\;\]
(see Figure~\ref{Abb:SterngitterStruktur}).  
For fixed $l$, the vertices $\v g_{ij}^l$ or $\v h_{ij}^l$ with $(i,j)\neq (0,0)$ of a grid mesh $\affr N$ build the $l$-th \emph{segment} of $\affr N$. 
The \emph{segment angle} of  $\affr N$  is $\varphi = 2 \pi f/m$.
The half-line from the center $\v g_{00}^{l}$ through $\v g_{10}^{l}$ is called the $l$-th \emph{spoke}, denoted by $S_l(\affr N)$ or $S_l$ for short.
\end{definition}

\begin{figure}[htb]
\begin{center}
{\includegraphics{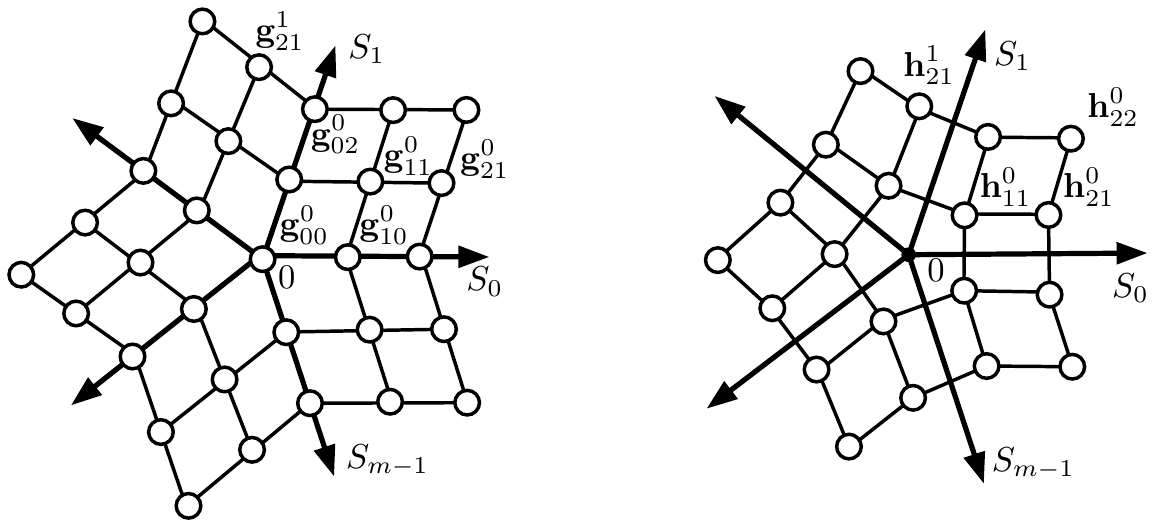}}
\caption{A primal grid mesh (left) and a dual grid mesh (right) with valence $5$ and frequency $1$.}
\label{Abb:SterngitterStruktur}
\end{center}
\end{figure}

Topologically, any ringnet $\affr M$ is equivalent to a grid mesh $\affr N$. Therefore, we use the same indices for equivalent vertices and denote the vertices of $\affr M$ by $\v p_{ij}^l$.  

\begin{definition}[{\rm Symmetric ringnet}]
A planar ringnet of valence~$m$ with vertices $\v p_{ij}^l$ in $\mathbb R^2$ is called \emph{rotation symmetric} with frequency $f$, if
\[\v p_{ij}^{l+1} = \left[\begin{array}{rr} \cos \theta & -\sin \theta\\ \sin \theta & \cos \theta\end{array} \right] \v p_{ij}^l\quad \mbox{with} \quad \theta = 2 \pi f/m\;.\]
A planar ringnet $\affr N \in \mathbb R^2$ is called \emph{reflection symmetric} if its permutation $\widetilde{\affr N}$ consisting of the points $\widetilde{\v p}_{ij}^l := \v p_{ji}^{(m-1)-l}$ equals the conjugate ringnet $\overline{\affr N}$ consisting of the points 
\[\overline{\v p_{ij}^l} = \overline{\left[\begin{array}{r} p_{ij, x}^l\\ p_{ij, y}^l \end{array}\right]} = \left[\begin{array}{r} p_{ij, x}^l\\ -p_{ij, y}^l \end{array}\right]\;,\]
i.\,e.,  
$\widetilde{\affr N} = \overline{\affr N}\;.$
A rotation and reflection symmetric ringnet is called \emph{symmetric}.
\end{definition}

\section{The characteristic mesh}\label{SECTION:ChapMap}
In this section, we construct a characteristic map of a general midpoint subdivision scheme 
\[U=A^{a_s} V^{v_s} R^{r_s} A^{a_{s-1}} V^{v_{s-1}} R^{r_{s-1}}  \cdots A^{a_1} V^{v_1} R^{r_1}\;,\]
where $\sum_{i=1}^s a_i + \sum_{i=1}^s v_i \ge 1$, $\sum_{i=1}^s r_i + \sum_{i=1}^s v_i \ge 1$,  and $\sum_{i=1}^s v_i$ is even.
We follow \cite{PC2011} and use results stated there for midpoint subdivision that are also valid for general midpoint subdivision since their proofs are only based on the properties that the subdivision scheme preserves symmetry and generates midpoints or any convex combinations.

\begin{theorem}[{\rm $\affr M_\infty$ and $\lambda_\varphi$}]
\label{SATZ:SUBDOMINANCE1}
Let $\affr M$ be the core mesh of a grid mesh with frequency $f$ and segment angle
\[ \varphi := \frac{2f \pi}{m} \in (0, \pi)\;.\]
Let
$\affr M_k := \left(U^k \affr M\right)_c / \left\Vert \left(U^k \affr M\right)_c\right\Vert\;,$
where $\Vert\cdot\Vert$ denotes any matrix norm. Then the following statements hold:
\begin{enumerate}
\item[(a)] The sequence $(\affr M_k)_{k \in \mathbb N}$ converges to a symmetric eigennet $\affr M_\infty$ with segment angle $\varphi$ and a positive eigenvalue $\lambda_\varphi$, which depends only on $\varphi$ but not on $f$ and $m$. 
$(\affr M_\infty)_{0\ldots 1}$ has at least one non-zero edge.
Additionally, we define $\lambda_\pi = |\gamma_\pi|$, where $\gamma_\pi$ is the maximum eigenvalue associated with a rotation symmetric eigenvector with segment angle~$\pi$.
\item[(b)] Restricting $U$ to the core meshes, the eigenvalue $\lambda_\varphi$ is the dominant eigenvalue of the eigenspaces of frequencies $f$ and $m-f$ and it has geometric and algebraic multiplicity $2$.
\item[(c)] $\lambda_\alpha > \lambda_\theta > \lambda_\pi \quad \mbox{for} \quad 0 < \alpha < \theta < \pi$.
\end{enumerate}
\end{theorem}

This can be proved as (5.4), (5.7), (6.3), and (6.4) in \cite{PC2011}. For midpoint subdivision schemes, $\lambda_\pi$ is equal to $1/4$ and to the subdominant eigenvalue $\mu_0$ of frequency $0$. This implies that $\lambda_{2\pi/m}$ is subdominant. However, for general midpoint subdivision schemes, $\lambda_{\pi}$ can be smaller than $\mu_0$. Therefore, we use the following lemma to show that $\lambda_{2\pi/m}$ is subdominant for $m\ge 5$.

\begin{lemma}[{\rm Scaling factor $\sigma$}]
\label{LEMMA:SUBDOMINANCEValence4}
\begin{enumerate}
\item[(a)] The scaling factor $\sigma = \sigma(U)$ defined in Equation~(\ref{EQ:SCALING}) is the subdominant eigenvalue of $U$ for $m=4$ and
\item[(b)] $\lambda_{\pi/2} = \sigma$ holds for any $m$ and $f$ such that $\frac{2f \pi}{m} = \frac{\pi}{2}$.
\end{enumerate}
\end{lemma}
\begin{proof}
{We consider a regular scalar-valued eigenmesh $\lambda \affr M = U \affr M$ with eigenvalue $\lambda$. Since
\[\lambda^2 \affr M_{\widehat\nabla \nabla} = \left(U^2\right)_{\widehat\nabla \nabla} \affr M_{\widehat\nabla \nabla}\]
and $\left\Vert  \left(U^2\right)_{\widehat\nabla \nabla}\right\Vert < \sigma \left(U^2\right) = \sigma^2$ (see the proof of Theorem~\ref{SATZ:C1REGULAR}), it follows that $|\lambda| < \sigma$ or that $\affr M_{\widehat\nabla \nabla} = 0$, meaning that $\affr M$ is a linear image of a regular grid $\affr G$, i.\,e., a linear combination of the constant mesh $[1\;\ldots\;1]^{\mathrm{t}}$ with eigenvalue $1$ and the two coordinates of $\affr G$. Since $U \affr G = \sigma \affr G$, (a) follows for $m=4$. Since there is a basis of rotation symmetric eigenmeshes, it suffices for $m\ne 4$ to consider a rotation symmetric mesh $\affr M$ with segment angle $\pi/2$. Due to symmetry, the subdivided mesh $U \affr M$ does not depend on  $f$, whence (b) follows.
}
\quad
\end{proof}

This lemma together with Theorem~\ref{SATZ:SUBDOMINANCE1} and Lemma~\ref{LEMMA:ZeilensummeVonB} can be used as in the proof of Theorem~(7.3) in \cite{PC2011} to derive the following corollary.

\begin{corollary}[{\rm Subdominant eigenvalue of $U$ for valencies $\ge 5$}]
\label{FOLGERUNG:SUBDOMINANCE2}
Let $\rho$ be as in Equation~(\ref{EQ:SubMatrix}), let $U$ be a general midpoint subdivision operator mapping the space of $c.\rho$-ringnets of valence $m$ to itself, and let $\affr M$ be a $c.\rho$-grid mesh of valence $m$ and frequency $1$.
If $m\ge 5$, the meshes
\[\affr M_k := \frac{U^k \affr M}{\left\Vert U^k \affr M\right\Vert}\]
converge to a subdominant eigenmesh $\affr M_\infty$ of $U$ called the \emph{characteristic mesh} of $U$ and its eigenvalue $\lambda_{2\pi/m}$ has geometric and algebraic multiplicity $2$.
\end{corollary}

\begin{remark}[{\rm Subdominant eigenvalue for valence $3$}]
\label{REMARK:SUBDOMINANCE}
Corollary~\ref{FOLGERUNG:SUBDOMINANCE2} is also true for $m=3$ if
\begin{equation}\label{EQ:m3}
\lambda_{2\pi/3}>\max \{|\mu_0|, \rho_B, \rho_A\}\;,
\end{equation} 
where $\rho_B$ and $\rho_A$ are the spectral radii defined in Lemma~\ref{LEMMA:ZeilensummeVonB}. 
\end{remark}

\begin{example}[{\rm The schemes $\mathit{VRV}$, $\mathit{VRVR}$, and $V^2$ for valence $3$}]
\label{EXAMPLE:Valence3}
Corollary~\ref{FOLGERUNG:SUBDOMINANCE2} holds for $U_1=\mathit{VRV}$ and $U_2=\mathit{VRVR}$, because 
$\mu_0 = \rho_B = \rho_A=0$, $\lambda_{2\pi/3}(U_1) = 1/8$, and $\lambda_{2\pi/3}(U_2) = 1/16$.
However, for $V^2$ we get 
$\lambda_{2\pi/3} = \mu_0 = \rho_B = 1/4$
and hence, Corollary~\ref{FOLGERUNG:SUBDOMINANCE2} does not apply.
\end{example}

\section{Smoothness for irregular meshes}\label{SECTION:GenMidC1}
In this section we analyze general midpoint schemes of the form 
\[U=(A^{a_s} V^{v_s} R^{r_s})  \cdots (A^{a_1} V^{v_1} R^{r_1}) = U_1 \cdots U_n\;,\]
where $\sum_{i} a_i + \sum_{i} v_i \ge 1 \le  \sum_{i} v_i + \sum_{i} r_i$  and $U_i \in \left\{A, R, V^2, \mathit{VAV}, \mathit{VA^\mathrm{2}V},¯\ldots\right\}$. We call these schemes \emph{$\mathrm{VAV}$-schemes} and analyze their characteristic maps by estimating their partial derivatives using cones defined by direction vectors of the spokes $S_0$ and $S_1$ of a grid mesh.

To simplify the notation, we identify the real plane $\mathbb R^2$ with the complex plane $\mathbb C$ by the bijection $\mathbb R^2 \ni [x\;y]^{\mathrm t} \mapsto x + \I y \in \mathbb C$.

Let
\[\affr C(\alpha, \beta)  := (0, \infty) e^{\I\, [\min(\alpha, \,\beta), \,\max(\alpha, \,\beta)]}\]
and
\[\affr C_0(\alpha, \beta)  := \affr C(\alpha, \beta) \cup \{0\}\]
be the unpointed and pointed cones spanned by $e^{\I \, \min(\alpha, \,\beta)}$ and $e^{\I \, \max(\alpha, \,\beta)}$, respectively.

Moreover, we define an operator $\nabla_2$ on a ringnet $\affr N = \left[\v p_{ij}^k\right]$ by
\begin{equation}\label{EQ:Nabla}
\nabla_2(\affr N) := \left\{\nabla_2 \v p^0_{i,j} = \v p^0_{i,j} - \v p^0_{i,j-1}  \;|\; i\ge 0, j>0\right\}\;.
\end{equation}

\begin{lemma}[{\rm A bound for $\nabla_2(\affr C)$}]
\label{LEMMA:AngleArea}
Let $U$ be a $\mathrm{VAV}$-scheme. Its characteristic mesh $\affr C = \left[\v c_{ij}^k\right]$ defined in Corollary~\ref{FOLGERUNG:SUBDOMINANCE2} and Remark~\ref{REMARK:SUBDOMINANCE} satisfies 
\[\nabla_2(\affr C) \subset \affr C\left(\frac{2\pi}{m}, \,\frac{\pi}{2}\right) =: \affr D\;.\]
\end{lemma}
\begin{proof}
{
Let $\affr M$ be a sufficiently large grid mesh with valence $m$ and frequency $1$.

First, we show that 
\[\nabla_2\left(U^k \affr M\right) \subset \affr D_0  := \affr C_0\left(\frac{2\pi}{m}, \,\frac{\pi}{2}\right)\]
holds for $k\ge 0$. 
Since $\nabla_2(\affr M) \subset \affr D_0$, it suffies to prove that $A$, $R$, and $\mathit{VA^lV}$ map a symmetric ringnet $\affr N$ with $\nabla_2(\affr N) \subset \affr D_0$ to a symmetric ringnet $\affr N'$ with $\nabla_2(\affr N') \subset \affr D_0$. Obviously, $\affr N'$ is symmetric. 
Furthermore,  the elements of $\nabla_2(A \affr N)$ and $\nabla_2(R \affr N)$ are linear combinations of elements in $\nabla_2(\affr N)$ with non-negative weights or, due to symmetry, are non-negative multiples of  $\v u_1$ or $\I \v u_0$, where $\v u_0:=1$ and $\v u_1:= e^{\I (2\pi / m)}$ are the direction vectors of the spokes $S_0$ and $S_1$, respectively.
For a net $A^lV\affr N$, we consider the $\v u_1$-diagonals of all faces belonging partly or completely to the $0$-th segment. The directions of these diagonals are either obtained by iteratively averaging the elements of $\nabla_2(\affr N)$ or are, due to symmetry, parallel to $\v u_1$ or $\I \v u_0$ (see the top and middle of Figure~\ref{Abb:MaskElements}). Hence, they lie in $\affr D_0$. Since every  element of $\nabla_2(VA^lV \affr N)$ is either obtained by halving the $\v u_1$-diagonals of $A^lV\affr N$ or, due to symmetry, is parallel to $\I \v u_0$ (see the bottom of Figure~\ref{Abb:MaskElements}), we conclude $\nabla_2(VA^lV \affr N) \subset \affr D_0$.

\begin{figure}[htb]
\begin{center}
{\includegraphics[scale=0.833333]{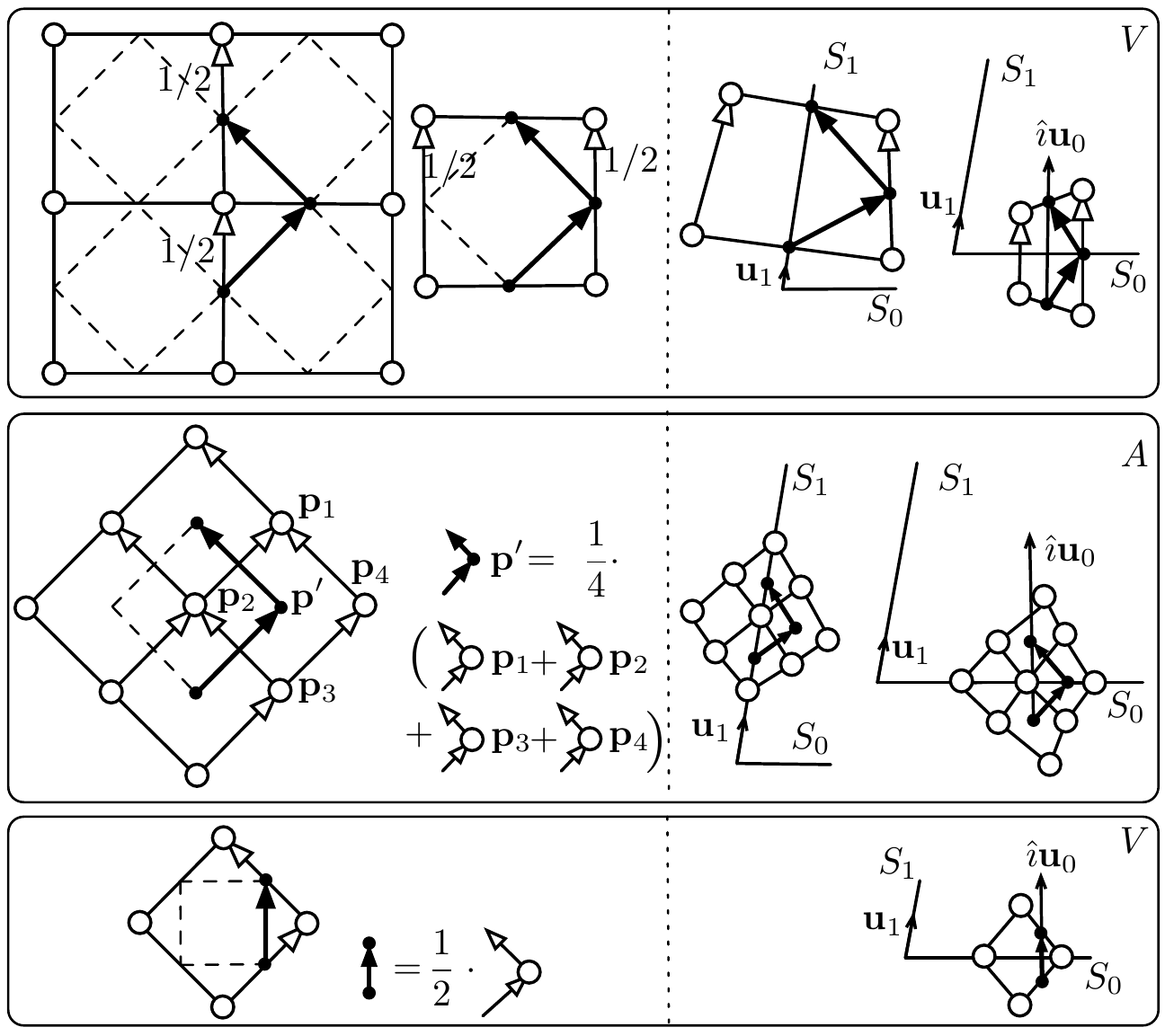}}
\caption{Generating $\v u_1$-diagonals by averaging and projecting for the operation $V$ (top) and for the operation  $A$ (middle) and generating $\nabla_2(VA^lV \cal{N})$ by halving and projecting the $\v u_1$-diagonals of $A^lV \cal{N}$ (bottom).}
\label{Abb:MaskElements}
\end{center}
\end{figure}

Second, we get $\nabla_2(\affr C) \subset \affr D_0$ because $\affr C=\lim_{k\to \infty}U^k \affr M/\left\Vert U^k \affr M \right\Vert$.

Finally, we show $\nabla_2(\affr C) \subset \affr D$. Due to Theorem~\ref{SATZ:SUBDOMINANCE1}\,(a), $\affr C_{0\ldots 1}$ has at least one non-zero edge. Let $\M S$ be the subdivision matrix of $U$ restricted to meshes of the same size as $\affr C$. For sufficiently large $k$, every element of $\nabla_2\left(S^k\affr C\right) \left(= \lambda^k \nabla_2(\affr C), \lambda>0\right)$ is a linear combination of $\nabla_2(\affr C)$ with non-negative weights and a positive weight for the non-zero edge in the $1$-ringnet. Hence, $\nabla_2(\affr C)$ has no zero elements and thus, $\nabla_2(\affr C) \subset \affr D$.
}
\qquad
\end{proof}

\begin{theorem}[{\rm $C^1$-property of $\mathrm{VAV}$-schemes}]
\label{SATZ:C1U}
Generic subdivision surfaces obtained by a $\mathrm{VAV}$-scheme $U$ are $C^1$ continuous at extraordinary points with valencies \mbox{$m\geq 5$}  and, if Inequality~(\ref{EQ:m3}) is satisfied, also with valence $m=3$ .
\end{theorem}
\begin{proof}
{
Let $\v c(x, y): \Omega \to \mathbb C$ be three segments of the characteristic map of $U$, where $\Omega = \Omega_{-1} \cup \Omega_{0} \cup \Omega_{1}$, as shown in Figure~\ref{Abb:CharMap-quad-2}, and $\v c_{\,|\, \Omega_i}$ is the $i$-th segment for $i=-1, 0, 1$.

\begin{figure}[h!]
\begin{center}
{\includegraphics{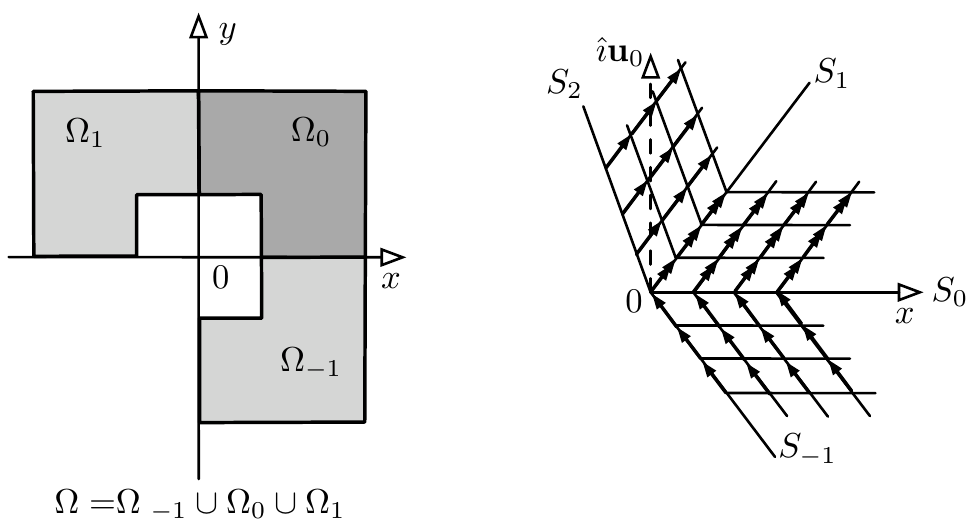}}
\caption{The domain $\Omega$ of $\v c$ (left) and the $-1$, $0$, $1$-th segments of a grid mesh (right), where the $y$-edges in these three segments are marked by arrows and the $y$-edges in the $0$-th segment are especially marked by double arrows.}
\label{Abb:CharMap-quad-2}
\end{center}
\end{figure}

First, we observe that for a symmetric ringnet $\affr N$, each element of $\nabla_2(A \affr N)$, $\nabla_2(2R \affr N)$, and $\nabla_2(2VA^lV \affr N)$ is a convex combination of elements in $\nabla_2(\affr N)$, in $\nabla_2(\affr N)$ reflected at $S_1$, and in $-\nabla_2(\affr N)$ reflected at $S_0$, where a reflected element has a weight which is less than or equal to that of the unreflected counterpart. Thus, by Lemma~\ref{LEMMA:AngleArea} and by induction, we see that $\nabla_2\left(U^k \affr C / \sigma^k\right) \subset \affr D_0$ for $k\ge 0$. Since every partial derivative $\v c_y$ over $\Omega_0$ is the limit of a sequence of vectors $\v v_k \in \nabla_2\left(U^k \affr C / \sigma^k\right)$, it follows that $\v c_y(\Omega_0) \subset \affr D_0$.

Second, we show $\v 0 \notin \v c_y(\Omega_0)$. Any $\v c_y(\v x)$, $\v x \in \Omega_0$, is a convex combination of  $\affr E_{-1}$ or $\affr E_0$, where $\affr E_i$ is the set of all $y$-edges in the segments $i$ and $i+1$ of $\affr U^k \affr C / \sigma^k$ for sufficiently large $k$. As in the proof of Lemma~\ref{LEMMA:AngleArea}, we observe 
\[\affr E_{-1} \subseteq  \affr C(2\pi/m, \, \pi-2\pi/m) \supseteq e^{\pi/2-2\pi/m} \affr E_0\;,\]
which implies $\v 0 \notin \v c_y(\Omega_0)$.

Hence, $\v c_y(\Omega_0) \subset \affr D$ and similarly, $\v c_x(\Omega_0) \subset \affr D-\pi/2$. Since each element of $\affr D$ is linearly independent with each element of $\affr D-\pi/2$, $\v c$ is regular over $\Omega_0$ (see Figure~\ref{Abb:Regularity-D}) and hence, the total characteristic map of $U$ is regular. 
Since $\v c_y(\Omega_0) \subset \affr D$, $\v c$ does not map any line segment between two points in $\Omega_0$ to a closed curve, meaning that $\v c$ is injective over $\Omega_0$.
Moreover, since $\affr M$ is a symmetric grid mesh whose $0$-th segment lies in  ${[}0,\; \infty ) \; e^{\I [0,\; 2\pi/m]} =: \affr A$ and $U$ preserves symmetry, it implies $\v c(\Omega_0) \subset \affr A$ and $\v c$ maps the interior of $\Omega_0$ into the interior of $\affr A$. Hence, the total characteristic map of $U$ is injective. 

\begin{figure}[h!]
\begin{center}
{\includegraphics{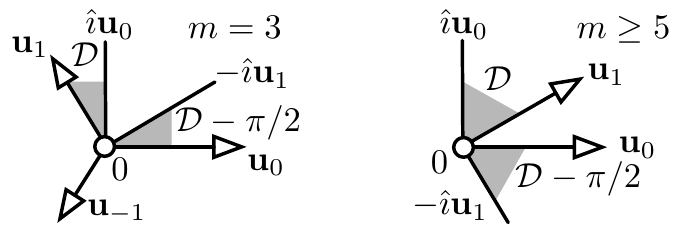}}
\caption{The direction cones $\affr D$ and $\affr D-\pi/2$.}
\label{Abb:Regularity-D}
\end{center}
\end{figure}

Thus, Reif's $C^1$-criterion \cite[Theorem 3.6]{Reif95} is satisfied, which finally concludes the theorem.
}
\qquad
\end{proof}

\begin{remark}[{\rm Checking the $C^1$-property for valence $3$}]
\label{REMARK:C1Um3}
According to Theorem~\ref{SATZ:C1U}, a $\mathrm{VAV}$-scheme generates $C^1$ surfaces at extraordinary points with valence $3$ if Inequality~(\ref{EQ:m3}) is satisfied.
Otherwise, if Inequality~(\ref{EQ:m3}) is not satisfied, there is no guarantee that the eigenvalue $\lambda_{2\pi/3}$ is a double subdominant eigenvalue and that it satisfies Reif's $C^1$-criterion as is the case with  $U=V^2$ (see Example~\ref{EXAMPLE:Valence3}).
In this case, we can use the more general $C^1$-criterion in \cite[Definition 5.10 and Theorem 5.25]{PR2008} with the weaker assumption that exactly two subdominant linearly independent generalized eigenvectors exist while the characteristic map defined by the (non-generalized) subdominant eigenvectors is regular and injective. Regularity and injectivity still follow (even for valence $3$) if $\nabla_2(\affr C) \subset \affr D_0$, where $\affr C$ is the control mesh of the characteristic map.
\end{remark}

\section{Two examples}\label{SECTION:Examples} 
In this section, we consider two general midpoint subdivision schemes and study their smoothness properties.

\begin{example}[{\rm $C^1$-property of $V^2$}]
\label{EXAMPLE:V2C1}
Since $V^2$ is a $\mathrm{VAV}$-scheme, it generates $C^1$ subdivision surfaces for regular meshes and for meshes with extraordinary elements of valencies $\ge 5$ according to Theorems~\ref{SATZ:C1REGULAR} and \ref{SATZ:C1U}.

For valence $3$, due to Example~\ref{EXAMPLE:Valence3}, Inequality~(\ref{EQ:m3}) is not satisfied for $V^2$.
According to \cite{PR97} and \cite[Section 6.3]{PR2008}, the subdivision matrix of $V^2$ for valence $3$ has exactly two linearly independent generalized eigenvectors associated with the subdominant eigenvalue $1/4$ and thus, the characteristic map exists. This map is regular and injective. This is shown in \cite{PR97} and can also easily be derived from Remark~\ref{REMARK:C1Um3}:

The $0$-th segment $\left[\v c_{ij}^0\right]$ of the control net $\affr C=\left[\v c_{ij}^k\right]$ is given by
\begin{eqnarray*}
\left[\v c_{ij}^0\right] 
&=& \left[ \begin{array}{rrr} 
\v c_{13}^0 & \v c_{23}^0 & \v c_{33}^0\\
\v c_{12}^0 & \v c_{22}^0 & \v c_{32}^0\\
\v c_{11}^0 & \v c_{21}^0 & \v c_{31}^0
\end{array}\right]
= \left[ \begin{array}{rrr}
-1 & 1.75 & 4\\ 0 & 2.5 & 4.75\\ 1 & 3 & 5
\end{array}\right]
+\sqrt{3}\I
\left[ \begin{array}{rrr}
3 & 3.75 & 4\\ 2 & 2.5 & 2.75\\ 1 & 1 & 1
\end{array}\right]\;.
\end{eqnarray*}
Thus, the edges $\nabla_2 \v c^0_{ij}$ $(j=3, 2, 1, \; i=1, 2, 3)$ of $\affr C$ are
\[
\begin{array}{rrr} 
-1+\sqrt{3}\I, & -0.75+1.25 \cdot \sqrt{3}\I, & -0.75+1.25 \cdot\sqrt{3}\I, \\
-1+\sqrt{3}\I, & -0.5+1.5 \cdot\sqrt{3}\I, & -0.25+1.75 \cdot\sqrt{3}\I,\\
2\cdot\sqrt{3}\I,  & 2\cdot\sqrt{3}\I, & 2\cdot\sqrt{3}\I.
\end{array}
\]
They all lie in $\affr D =  \affr C\left(\frac{\pi}{2}, \,\frac{2\pi}{3}\right)$.
Hence, due to Remark~\ref{REMARK:C1Um3}, the scheme $V^2$ also generates $C^1$ surfaces for valence $3$.
\end{example}

In general, if a general midpoint scheme is not a $\mathrm{VAV}$-scheme, the $C^1$ analysis technique in this paper usually does not apply.

\begin{example}[{\rm Non $\mathrm{VAV}$-schemes}]
\label{EXAMPLE:VAV}
We consider the non $\mathrm{VAV}$-scheme $\mathit{VRVR}$. 
It generates $C^1$ subdivision surfaces for regular meshes according to Theorem~\ref{SATZ:C1REGULAR}.

Operate this scheme on a dual grid mesh $\affr M$ with segment angle $\varphi$. In Figure~\ref{FIG:VRVRVRVR}, the edge $\v e$ in $(\mathit{VRVR})^2 \affr M$ has the  direction angle $\pi-\arctan(16\tan \varphi)$ and hence, it does not lie in the pointed cone $\affr D_0$ spanned by $e^{\I \varphi}$ and $e^{\I \pi/2}$.

Hence, for this scheme, we cannot check the regularity and the injectivity of the characteristic map by using the $C^1$ analysis technique in this paper (see the proof of Theorem~\ref{SATZ:C1U}).
\end{example}

\begin{figure}[h!]
\begin{center}
{
\includegraphics[scale=0.833333]{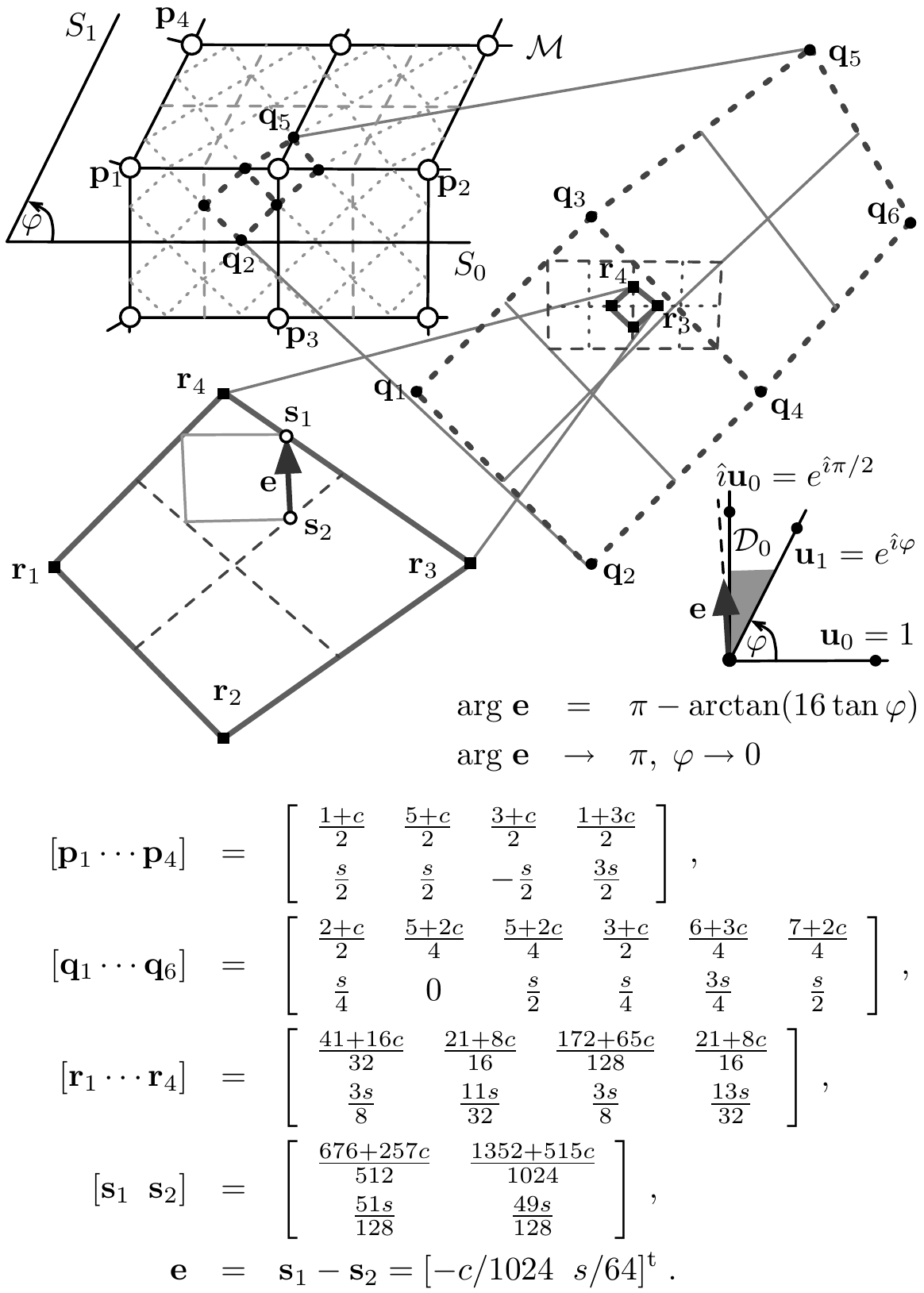}
}
\caption{Since $\nabla_2\left(U^2 \affr M\right) \ni \v e \notin \affr D_0$\;, $\nabla_2\left(U^k\affr M/\sigma^k\right) \subset \affr D$ does not hold, where $\affr M$ is a dual grid mesh with segment angle $\varphi$, $c=\cos \varphi$, $s=\sin \varphi$, and $U=\mathit{VRVR}$\;.}
\label{FIG:VRVRVRVR}
\end{center}
\end{figure}

\section{Generalized Catmull-Clark subdivision}\label{SECTION:DefGenMidSubII}

In this section, we introduce the \emph{generalized Catmull-Clark subdivision schemes}, which are another generalization of the midpoint subdivision schemes, and analyze the smoothness of the subdivision surfaces.

For an odd degree $n=2r+1\ge3$, we generalize the midpoint scheme
\[M_n = A^{n-1}R = A^2 \,\cdots\, A^2 \, R
\quad
\mbox{to}
\quad
\overline{M}_n = B_{r}  \,\cdots \,B_1\, R,\]
and, for an even degree $n=2r+2\ge 2$, we generalize the midpoint scheme
\[M_n = A^{n-1}R = A\, A^2 \, \cdots A^2 \, R
\quad
\mbox{to}
\quad
\overline{M}_n = A\, B_{r} \, \cdots B_{1}\, R,\]
where $R$ is the refinement operator shown at the top of Figure~\ref{FIG:OperatorsRAV} and where each smoothing operator \mbox{$B_i=B_{\alpha_i,\beta_i}$} is a generalized operator of $A^2$ with the mask shown in Figure~\ref{Abb:MaskB}. The functions $\alpha_i(\cdot)$ and $\beta_i(\cdot)$ are non-negative functions depending on $i$ and they satisfy $0<\alpha_i + \beta_i < 1$, $\alpha_i(4) = 1/4$, and $\beta_i(4) = 1/2$. 
\begin{figure}[htb]
\begin{center}
{\includegraphics{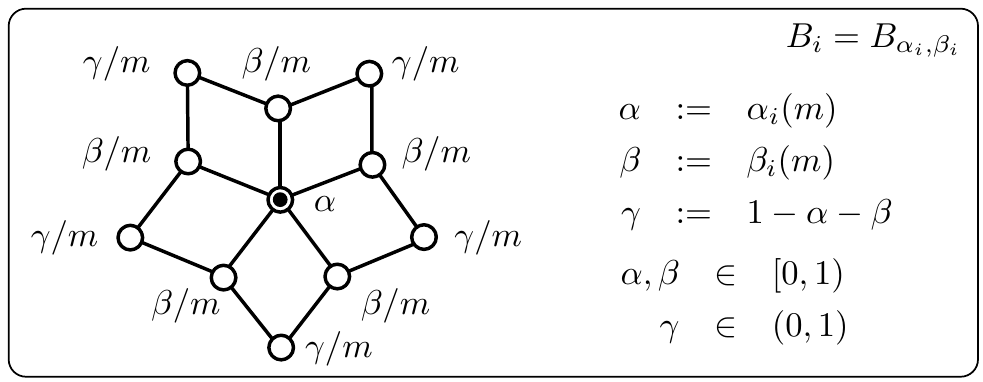}}
\caption{Smoothing operator $B_i=B_{\alpha_i,\beta_i}$ for a vertex of valence $m$.}
\label{Abb:MaskB}
\end{center}
\end{figure}

It is straightforward to see that $B_1 = A^2$ holds on regular meshes and that $\overline{M}_3 = B_1\, R$ is the Catmull-Clark scheme with restricted parameters $\alpha_1$ and $\beta_1$,  i.\,e.,  at extraordinary points with valencies $m(\ne 4)$, $\alpha_1$ and $\beta_1$ satisfy $\alpha_1(m), \beta_1(m) \in [0,\,1)$ and  $\alpha_1(m)+ \beta_1(m) \in (0,\,1)$.

If $\alpha_1 \equiv 1/4$ and $\beta_1 \equiv 1/2$, then $B_1 = A^2$ holds on arbitrary meshes and $\overline{M}_n = B_1^r \,R$ or $\overline{M}_n = A\, B_1^r \,R$ is a midpoint subdivision scheme of degree $n$ if $n$ is odd or even, respectively.

\begin{theorem}[{\rm $C^1$-property of $\overline{M}_n$}]
\label{SATZ:GenMidSubII}
Let 
$\overline{M}_n = B_{r}\cdots B_1 R$ or $\overline{M}_n = B_{r} \cdots B_1 R$
be a scheme as above with $n\ge 2$. Let $\lambda$ be the dominant eigenvalue associated with frequency $1$ and let $\mu_0$ be the subdominant eigenvalue associated with frequency $0$.
If $\lambda>|\mu_0|$ for all valencies $m\ge 5$ or $m=3$, then $\overline{M}_n$ is a $C^1$-scheme.
\end{theorem}
\begin{proof}
First, for regular meshes, $\overline{M}_n$ is the Lane-Riesenfeld scheme of degree $n$ and thus, it is a $C^1$-scheme.

Second, at extraordinary points, $\overline{M}_n$ is convergent because, as in the proof of Theorem~\ref{SATZ:C0}, 
it can easily be verified that the subdivision matrix of $\overline{M}_n$ is stochastic with the simple dominant eigenvalue $1$. 

Next, let $\affr M$ be a rotation symmetric  ringnet with non-zero frequency. If $n$ is odd, we require $\affr M$ to be primal and otherwise to be dual. For any choice of the parameters $\alpha_i$ and $\beta_i$ in $B_i$, we get $B_i \affr N = A^2 \affr N$ for any primal rotation symmetric ringnet $\affr N$ with non-zero frequency and thus $\overline{M}_n \affr M = M_n \affr M$. Hence, ${M}_n$ and $\overline{M}_n$ share the same rotation symmetric eigenvectors with non-zero frequency and the same associated eigenvalues. Using Theorem~\ref{SATZ:SUBDOMINANCE1} and the assumption $\lambda>|\mu_0|$, we get that $\overline{M}_n$ has the double subdominant eigenvalue $\lambda = \lambda_{2\pi/m}$. Hence, $\overline{M}_n$ and $M_n$ have the same characteristic map. 
For the midpoint scheme $M_n$,  we have $|\mu_0(M_n)| = \lambda_\pi (M_n)= 1/4 < \lambda_{2\pi/m} (M_n)$ and, according to Lemma~\ref{LEMMA:ZeilensummeVonB}, also $\rho_B(M_n), \rho_A(M_n) \le 1/4$. Thus, Inequality~(\ref{EQ:m3}) is satisfied and moreover, as in the proof of Theorem~\ref{SATZ:C1U}, the characteristic map is regular and injective for all valencies $m\ge 5$ or $m=3$.  (One can also see \cite[(8.1) Theorem]{PC2011} for a proof of the regularity and injectivity of the characteristic map.) Thus, Reif's $C^1$-criterion \cite[Theorem 3.6]{Reif95} is satisfied, which concludes the theorem.
\qquad
\end{proof}

\begin{remark}[{\rm Generalized Catmull-Clark subdivision}]
\label{REM:GeneralizedCC}
The schemes $\overline{M}_n$ can be further generalized by giving up the conditions $\alpha_i(4) = 1/4$ and $\beta_i(4) = 1/2$. This generalized subdivision is called \emph{generalized Catmull-Clark subdivision} and the resulting subdivision surfaces are analyzed in detail in the technical report \cite{CP2012-GenCC}, where we show that the generalized Catmull-Clark subdivision surfaces are $C^1$ at all regular points and at extraordinary points if one of the following conditions is satisfied: 
\begin{enumerate}
\item[(1)] $m\geq 5$ and $\alpha_i, \beta_i$ are constant functions.
\item[(2)] $m\geq 5$, $\alpha_i, \beta_i$ are non-constant functions, and $\lambda> |\mu_0|$.
\item[(3)] $m=3$ and Inequality~(\ref{EQ:m3}) is satisfied, i.\,e., $\lambda_{2\pi/m}>\max \{|\mu_0|, \rho_B, \rho_A\}$.
\end{enumerate}
Here, $m$ is the valence of the extraordinary element in the control mesh and $\lambda$ and $\mu_0$ are defined in the above theorem.
\end{remark}

\section{Conclusion}\label{SECTION:Conclusion}

In this paper, two generalizations of midpoint subdivision, general midpoint subdivision and generalized Catmull-Clark subdivision, are introduced. They build two infinite classes of subdivision schemes, where the first class includes the midpoint subdivision schemes and the mid-edge, or simplest, subdivision scheme and the second class includes the midpoint subdivision schemes and the Catmull-Clark subdivision scheme with restricted parameters.

General midpoint subdivision surfaces for regular meshes are analyzed by estimating the norm of a special second order difference scheme and they are all $C^1$ continuous.
For irregular meshes, the smoothness is analyzed by a $C^1$ analysis technique developed in \cite{PC2011} for midpoint subdivision.  For this technique, we worked out several adaptions to the situation considered in this paper. In particular, the properties of the characteristic maps are analyzed without explicit knowledge of these characteristic maps and a deeper understanding of the spectral properties of the subdivision matrices at extraordinary points is provided. 

We see the results of this paper as a step towards even more general smoothness results in the spirit of the smoothness characterizations known for univariate corner cutting schemes \cite{GQ96,PPS97}.
In a consequent paper \cite{CP2012-GenTriMid}, the established $C^1$ analysis technique is generalized to apply to specific infinite \emph{triangular} subdivision schemes.
We hope that our $C^1$ analysis technique can be applied  to other subdivision schemes as well, e.\,g. schemes that can be factorized into general and simple convex combination operators.


\bibliographystyle{alpha}
\bibliography{../../Literatur/CAGDLiteratur}

\end{document}